\documentclass[11pt]{article}
\def\comments{0}

\usepackage[utf8]{inputenc}
\usepackage{graphicx}
\usepackage{amsmath, bbm}
\usepackage{amsfonts}
\usepackage{amsthm}
\usepackage[margin=1.0in]{geometry}
\usepackage[ruled,vlined]{algorithm2e}
	\SetKwFor{While}{While}{}{}
	\SetKwFor{For}{For}{}{}
	\SetKwIF{If}{ElseIf}{Else}{If}{}{Else If}{Else}{}
	\SetKw{Return}{Return}
\usepackage{algorithmicx}
\usepackage[english]{babel}
\usepackage{microtype}
\usepackage{verbatim}
\usepackage{wrapfig}
\usepackage{authblk}

\usepackage[dvipsnames]{xcolor}
\usepackage{hyperref}

\usepackage{caption,subcaption}
\theoremstyle{definition}

\newtheorem{definition}{Definition}[section]
\newtheorem{theorem}{Theorem}
\newtheorem{lemma}[theorem]{Lemma}

\DeclareMathOperator*{\argmax}{arg\,max}
\DeclareMathOperator*{\argmin}{arg\,min}

\let\originalleft\left
\let\originalright\right
\renewcommand{\left}{\mathopen{}\mathclose\bgroup\originalleft}
\renewcommand{\right}{\aftergroup\egroup\originalright}

\renewcommand{\epsilon}{\varepsilon}
\newcommand{\eps}{\varepsilon}
\newcommand{\epsth}{\eps_{\mathit{th}}}
\newcommand{\epslb}{\eps_{\mathit{LB}}}
\newcommand{\epsopt}{\eps_{OPT}}
\renewcommand{\log}{\ln}

\ifnum\comments=1
    \newcommand{\matthew}[1]{\marginpar{\tiny\color{orange}{MJ: #1}}}
    \newcommand{\jon}[1]{\marginpar{\tiny\color{blue}{JU: #1}}}
    \newcommand{\alina}[1]{\marginpar{\tiny\color{purple}{AO: #1}}}
\else
    \newcommand{\matthew}[1]{}
    \newcommand{\jon}[1]{}
    \newcommand{\alina}[1]{}
\fi

\newcommand{\mypar}[1]{\medskip\noindent\textbf{#1}}

\title{Auditing Differentially Private Machine Learning: \\ How Private is Private SGD?\thanks{Authors ordered by contribution.}}

\author{Matthew Jagielski}
\author{Jonathan Ullman}
\author{Alina Oprea}
\affil{Khoury College of Computer Sciences, Northeastern University}
\date{}

\begin{document}
\maketitle

\begin{abstract}
We investigate whether Differentially Private SGD offers better privacy in practice than what is guaranteed by its state-of-the-art analysis.  We do so via novel data poisoning attacks, which we show correspond to realistic privacy attacks.  While previous work (Ma et al., arXiv 2019) proposed this connection between differential privacy and data poisoning as a defense against data poisoning, our use as a tool for understanding the privacy of a specific mechanism is new.  More generally, our work takes a quantitative, empirical approach to understanding the privacy afforded by specific implementations of differentially private algorithms that we believe has the potential to complement and influence analytical work on differential privacy. An open-source implementation of our algorithms can be found at \url{https://github.com/jagielski/auditing-dpsgd}.
\end{abstract}

\section{Introduction}
Differential privacy~\cite{DworkMNS06} has become the de facto standard for guaranteeing privacy in machine learning and statistical analysis, and is now being deployed by many organizations including Apple~\cite{AppleDP}, Google~\cite{ErlingssonPK14,Bittau+17,Papernot+18}, and the US Census Bureau~\cite{Haney+17}.  Now that differential privacy has moved from theory to practice, there has been considerable attention on optimizing and evaluating differentially private machine learning algorithms, notably differentially private stochastic gradient descent (henceforth, DP-SGD)~\cite{SongCS13,BassilyST14,Abadi+16}, which is now widely available in TensorFlow Privacy~\cite{TFP}. DP-SGD is the building block for training many widely used private classification models, including feed-forward and convolutional neural networks.

Differential privacy gives a strong \emph{worst-case} guarantee of individual privacy: a differentially private algorithm ensures that, for any set of training examples, no attacker, no matter how powerful, can learn much more information about a single training example than they could have learned had that example been excluded from the training data.  The amount of information is quantified by a \emph{privacy parameter} $\eps$.\footnote{There are several common variants of differential privacy~\cite{DworkKMMN06,DworkR16,BunS16,Mironov17,BunDRS18,DongRS19} that quantify the influence of a single example in slightly different ways, sometimes using more than one parameter.  For this high-level discussion, we focus on the single, primary parameter $\eps$.}  Intuitively, a smaller $\eps$ means stronger privacy protections, but leads to lower accuracy.  As such there is often pressure to set this parameter as large as one feels still gives a reasonable privacy guarantee, and relatively large parameters such as $\eps = 2$ are not uncommon.  However, this guarantee is not entirely satisfying, as such an algorithm might allow an attacker to guess a random bit of information about each training example with approximately 86\% accuracy.  As such there is often a gap between the strong formal protections promised by differential privacy and the specific quantitative implications of the choice of $\eps$ in practice.

This state-of-affairs is often justified by the fact that our analysis of the algorithm is often pessimistic.  First of all, $\eps$ is a parameter that has to be determined by careful analysis, and often existing theoretical analysis is not tight.  Indeed a big part of making differentially private machine learning practical has been the significant body of work giving progressively more refined privacy analyses specifically for DP-SGD~\cite{Abadi+16,DongRS19,MironovTZ19,YuLPGT19}, and for all we know these bounds on $\eps$ will continue to shrink.  Indeed, it is provably intractable to determine the tightest bound on $\eps$ for a given algorithm~\cite{GilbertM18}. Second, differential privacy is a worst-case notion, as the mechanism might have stronger privacy guarantees on realistic datasets and realistic attackers.  Although it is plausible that differentially private algorithms with large values of $\eps$ provide strong privacy in practice, it is far from certain, which makes it difficult to understand the appropriate value of $\eps$ for practical deployments.

\subsection{Our Contributions}
\mypar{Auditing DP-SGD.} In this paper we investigate the extent to which DP-SGD,\footnote{Although our methods are general, in this work we exclusively study the implementation and privacy analysis of DP-SGD in TensorFlow Privacy~\cite{TFP}.} does or does not give better privacy in practice than what its current theoretical analysis suggests. We do so using novel \emph{data poisoning attacks}.  Specifically, our method starts with a dataset $D$ of interest (e.g.~Fashion-MNIST) and some algorithm $\mathcal{A}$ (e.g.~DP-SGD with a specific setting of hyperparameters), and produces a small poisoning set $S$ of $k$ points and a binary classifier $T$ such that $T$ distinguishes the distribution $\mathcal{A}(D)$ from $\mathcal{A}(D \cup S)$ with significant advantage over random guessing.  If $\mathcal{A}$ were $\eps$-DP, then $T$ could have accuracy at most $\exp(\eps k) / (1+\exp(\eps k)) $, so if we can estimate the accuracy of $T$ we can infer a lower bound on $\eps$.
While previous work~\cite{ma2019data} proposed to use this connection between differential privacy and data poisoning as a defense against data poisoning, our use in this context of auditing the privacy of DP-SGD is new.

Specifically, for certain natural choices of hyperparameters in DP-SGD, and standard benchmark datasets (see Figure~\ref{fig:images}), our attacks give lower bounds on $\eps$ that are approximately 10x better than what we could obtain from previous methods, and are within approximately 10x of the worst-case, analytically derived upper bound. 
For context, previous theoretical improvements to the analysis have improved the worst-case upper bounds by factors of more than 1000x over the na\"ive analysis, and thus our results show that we cannot hope for similarly dramatic gains in the future.

\mypar{Novel Data Poisoning Attacks.}
We find that existing data poisoning attacks, as well as membership inference attacks proposed by prior work, have poor or trivial performance not only against DP-SGD, but even against SGD with gradient clipping (i.e.~rescaling gradients to have norm no larger than some $C$).  Gradient clipping is an important part of DP-SGD, but does not provide any formal privacy guarantees on its own.  Thus, we develop a novel data poisoning attack that is more robust to gradient clipping, and also performs much better against DP-SGD.



Intuitively, data poisoning attacks introduce new points whose gradients will change the model in a certain direction, and the attack impact increases when adding poisoning points of larger gradients.  Existing attacks modify the model in a random direction, and have to push far enough that the original distribution on model parameters and the new distribution become distinguishable.  To be effective, these attacks use points which induce large gradients, making the attack sensitive to gradient clipping. On the other hand, our attack improves by finding the direction where the model parameters have the lowest variance, and select poisoning points that modify the model in that direction.  Therefore, we achieve the same effect of model poisoning with poisoning points of smaller gradients, thereby making the attack more robust to clipping.

\mypar{The Role of Auditing in DP.} More generally, our work takes a quantitative, empirical approach to \emph{auditing} the privacy afforded by specific implementations of differentially private algorithms.  We do not advocate trying to definitively measure privacy of an algorithm empirically, since it's hopeless to try to anticipate all future attacks.  Rather, we believe this empirical approach has the potential to complement and influence analytical work on differential privacy, somewhat analogous to the way \emph{cryptanalysis} informs the design and deployment of \emph{cryptography}.

Specifically, we believe this approach can complement the theory in several ways: 
\begin{itemize}
    \item Most directly, by advancing the state-of-art in privacy attacks, we can either demonstrate that a given algorithm with a given choice of parameters is not sufficiently private, or give some confidence that it might be sufficiently private.
    \item Establishing strong lower bounds on $\eps$ gives a sense of how much more one could hope to get out of tightening the existing privacy analysis.
    \item Observing how the performance of the attack depends on different datasets, hyperparameters, and variants of the algorithm can identify promising new phenomena to explore theoretically.
    \item Producing concrete privacy violations can help non-experts interpret the concrete implications of specific choices of the privacy parameter.
\end{itemize}

\subsection{Related Work}

\mypar{DP-SGD.} Differentially private SGD was introduced in~\cite{SongCS13}, and an asymptotically optimal analysis of its privacy properties was given in~\cite{BassilyST14}.  Notably Abadi et al.~\cite{Abadi+16} gave greatly improved concrete bounds on its privacy parameter, and showed its practicality for training neural networks, making DP-SGD one of the most promising methods for practical private machine learning.  There have been several subsequent efforts to refine the privacy analysis of this specific algorithm~\cite{MironovTZ19,DongRS19,YuLPGT19}.  A recent work~\cite{HylandT19} gave a heuristic argument that SGD \emph{itself} (without adding noise to the gradients) satisfies differential privacy, but even then the bounds on $\eps$ are quite large (e.g.~$\eps = 13.01$) for most datasets. 



\mypar{Privacy Attacks.}
Although privacy attacks have a very long history, the history of privacy attacks against aggregate statistical information, such as machine learning models, goes back to the seminal work of Dinur and Nissim~\cite{DinurN03} on \emph{reconstruction attacks}.  A similar, but easier to implement type of attack, \emph{membership inference attacks}, was first performed by Homer et al.~\cite{homer2008resolving}, and theoretically analyzed in~\cite{SankararamanOJH09,DworkSSUV15}.  Shokri et al.~\cite{shokri2017membership} and Yeom et al.~\cite{yeom2018privacy} gave black-box membership inference algorithms for complex machine learning models.  Membership inference attacks are compelling because they require relatively weak assumptions, but, as we show, state-of-the-art membership inference attacks lead to quantitatively weak privacy violations.

More directly related to our work, privacy attacks were recently used by Jayaraman and Evans~\cite{jayaraman2019evaluating} to understand the concrete privacy leakage from differentially private machine learning algorithms, specifically DP-SGD.  However, the goal of their work is to compare the privacy guarantees offered by different variants of differential privacy, rather than to determine the level of privacy afforded by a given algorithm.  As such, their attacks are quantitatively much less powerful than ours (as we show in Figure~\ref{fig:images}), and are much further from determining the precise privacy guarantees of DP-SGD. 



\mypar{Differential Privacy and Data Poisoning.}
Ma et al.~\cite{ma2019data} and Hong et al.~\cite{hong2020effectiveness} evaluate the effectiveness of data poisoning attacks on differentially private machine learning algorithms. Ma et al. consider both the output perturbation and objective perturbation algorithms for learning ridge regression and logistic regression models, proposing attacks on differentially private algorithms, and also argue that differentially privacy can serve as a defense for poisoning attacks. Hong et al. propose differential privacy as a defense for poisoning attacks, showing that DP-SGD performs effectively at defending against existing poisoning attacks in the literature.  While differential privacy provides a provable defense for poisoning attacks, our intuition is that the strong poisoning attacks we design allow us to measure a lower bound on the privacy offered by differentially private algorithms.



\mypar{Automated Discovery of Privacy Parameters.}  Two works have focused on automatically discovering (upper or lower bounds on) privacy parameters.  \cite{GilbertM18} showed that determining the exact privacy level using black-box access to the algorithm is prohibitively expensive.  In the white-box setting, Ding et al.~\cite{DingWWZK18} used a clever combination of program analysis and random sampling to identify violations of $\eps$-DP, although their methods are currently limited to simple algorithms.  Moreover the violations of DP they identify may not correspond to realistic attacks.

\section{(Measuring) Differential Privacy}
\subsection{Differential Privacy Background}
We begin by outlining differential privacy and one of its relevant properties: group privacy. We consider machine learning classification tasks, where a dataset consists of many samples from some domain $\mathcal{D}=\mathcal{X}\times\mathcal{Y}$, where $\mathcal{X}$ is the feature domain, and $\mathcal{Y}$ the label domain. We say two datasets $D_0, D_1$ differ on $k$ rows if we can replace at most $k$ elements from $D_0$ to produce $D_1$.

\begin{definition}[\cite{DworkMNS06}]
An algorithm $\mathcal{A}: \mathcal{D}\mapsto\mathcal{R}$ is $(\eps, \delta)$-\emph{differentially private} if for any two datasets $D_0, D_1$ which differ on at most one row, and every set of outputs $\mathcal{O} \subseteq \mathcal{R}$:
\begin{equation}
\Pr[\mathcal{A}(D_0)\in \mathcal{O}]\le e^{\eps}\Pr[\mathcal{A}(D_1)\in\mathcal{O}]+\delta,
\label{eq:dp}
\end{equation}
where the probabilities are taken only over the randomness of $\mathcal{A}$.
\end{definition}

\begin{lemma}[Group Privacy]
Let $D_0, D_1$ be two datasets differing on at most $k$ rows, $\mathcal{A}$ is an $(\eps, \delta)$-differentially private algorithm, and $\mathcal{O}$ an arbitrary output set. Then
\begin{equation}
\Pr[\mathcal{A}(D_0)\in \mathcal{O}]\le e^{k\eps}\Pr[\mathcal{A}(D_1)\in\mathcal{O}]+ \tfrac{e^{k\eps}-1}{e^{\eps}-1} \cdot \delta.
\label{eq:group}
\end{equation}
\end{lemma}

Group privacy will give guarantees for poisoning attacks that introduce multiple points.

\mypar{DP-SGD.}
The most prominent differentially private mechanism for training machine learning models is differentially private stochastic gradient descent (DP-SGD)~\cite{SongCS13,BassilyST14,Abadi+16}. DP-SGD makes two modifications to the standard SGD procedure: 
\begin{figure}
\vspace{-5mm}
\begin{algorithm}[H]
\SetAlgoLined
\KwData{Input: Clipping norm $C$, noise magnitude $\sigma$, iteration count $T$, batch size $b$, dataset $D$, initial model parameters $\theta_0$, learning rate $\eta$}
\For{$i\in[T]$}{
  $G=0$
  
  \For{$(x,y)\in$ batch of $b$ random elements of $D$}{
   $g=\nabla_{\theta}\ell(\theta_i; (x, y))$
   
   $G = G + b^{-1}g\cdot\min(1, C||g||_2^{-1})$
   
   }{
   $\theta_i=\theta_{i-1} - \eta (G+\mathcal{N}(0, (C\sigma)^2\mathbb{I}))$
  }
 }
 \Return $\theta_T$
 \caption{DP-SGD}
 \label{alg:dpsgd}
\end{algorithm}
\vspace{-5mm}
\end{figure}
clipping gradients to a fixed maximum norm $C$, and adding noise to gradients with standard deviation $\sigma C$, for a given $\sigma$, as shown in Algorithm~\ref{alg:dpsgd}. Given the hyperparameters -- clipping norm, noise magnitude, iteration count, and batch size -- one can analyze DP-SGD to conclude that it satisfies $(\eps,\delta)$-differential privacy for some parameters $\eps,\delta \geq 0$.


\subsection{Statistically Measuring Differential Privacy}
In this section we describe our main statistical procedure for obtaining lower bounds on the privacy parameter for a given algorithm $\mathcal{A}$, which functions differently from the membership inference attack used in prior work (\cite{shokri2017membership, jayaraman2019evaluating} and described in Appendix
\ref{app:mi}).
Here, we describe the procedure generally, in the case where $\delta=0$; in Appendix
\ref{app:deltage0},
we show how to adapt the procedure for $\delta>0$, and in Section~\ref{sec:poisoning}, we discuss how we instantiate it in our work. Given a learning algorithm $\mathcal{A}$, we construct two datasets $D_0$ and $D_1$ differing on $k$ rows, and some output set $\mathcal{O}$. We defer the discussion of constructing $D_0$, $D_1$, and $\mathcal{O}$ to Section~\ref{sec:poisoning}. We also wish to bound the probability that we incorrectly measure $\epslb$ by a small value $\alpha$. From Equation~(\ref{eq:group}), observe that by estimating the quantities $p_0 = \Pr[\mathcal{A}(D_0)\in\mathcal{O}]$ and $p_1 = \Pr[\mathcal{A}(D_1)\in\mathcal{O}]$, we can compute the largest $\epslb$ such that Equation~(\ref{eq:group}) holds. With $\delta=0$, this simplifies to $\epslb=\ln(p_0/p_1)$. This serves as an estimate of the leakage of the private algorithm, but requires estimating $p_0$ and $p_1$ accurately.

For an arbitrary algorithm, it's infeasible to compute $p_0, p_1$ precisely, so we rely on Monte Carlo estimation, by training some fixed $T$ number of times. However, this approach incurs statistical uncertainty, which we correct for by using Clopper Pearson confidence intervals~\cite{clopper1934use}. That is, to ensure that our estimate $\epslb$ holds with probability $>1-\alpha$, we find a Clopper Pearson lower bound for $p_1$ that holds with probability $1-\alpha/2$, and an upper bound for $p_0$ holding with probability $1-\alpha/2$. Qualitatively, we can be confident that our lower bound on privacy leakage $\epsilon'$ holds with probability $1-\alpha$. This procedure is outlined in Algorithm~\ref{alg:generic}, and we prove its correctness in Theorem~\ref{thm:generic}.

\begin{algorithm}[H]
\DontPrintSemicolon
\SetAlgoLined
\KwData{Algorithm $\mathcal{A}$, datasets $D_0$, $D_1$ at distance $k$, output set $\mathcal{O}$, trial count $T$, confidence level $\alpha$}
\vspace{2pt}
$\text{ct}_0 = 0, \text{ct}_1 = 0$

\For{$i\in[T]$}{
  \lIf{$\mathcal{A}(D_0)\in\mathcal{O}$}{
   $\text{ct}_0 = \text{ct}_0 + 1$
   }
  \lIf{$\mathcal{A}(D_1)\in\mathcal{O}$}{
   $\text{ct}_1 = \text{ct}_1 + 1$
   }
 }
 $\hat{p}_0 = \textsc{ClopperPearsonLower}(\text{ct}_0, T, \alpha/2)$
 
 $\hat{p}_1 = \textsc{ClopperPearsonUpper}( \text{ct}_1, T, \alpha/2)$
 
 \Return $\epslb = \ln(\hat{p}_0/\hat{p}_1)/k$
 \caption{Empirically Measuring $\eps$}
 \label{alg:generic}
\end{algorithm}

\begin{theorem} \label{thm:generic}
When provided with black box access to an algorithm $\mathcal{A}$, two datasets $D_0$ and $D_1$ differing on at most $k$ rows, an output set $\mathcal{O}$, a trial number $T$ and statistical confidence $\alpha$, if Algorithm~\ref{alg:generic} returns $\epslb$, then, with probability $1-\alpha$, $\mathcal{A}$ does not satisfy $\eps'$-DP for any $\eps' < \epslb$.
\end{theorem}
We stress that when we say $\eps_{\mathit{LB}}$ is a lower bound with probability $1-\alpha$, this is only over the randomness of the Monte Carlo sampling, and is not based on any modeling or assumptions.  We can always move our confidence closer to $1$ by taking $T$ larger.

\begin{proof}[Proof of Theorem~\ref{thm:generic}]
First, the guarantee of the Clopper-Pearson confidence intervals is that, with probability at least $1-\alpha$, $\hat{p}_0 \le p_0$ and $\hat{p}_1 \ge p_1$, which implies $p_0 / p_1 \geq \hat{p}_0 / \hat{p}_1$.  Second, if $\mathcal{A}$ is $\eps$-DP, then by group privacy we would have $p_0/p_1 \leq \exp(k\eps)$, meaning $\mathcal{A}$ is \emph{not} $\eps'$-DP for any $\eps' < \frac{1}{k} \log(p_0/p_1)$.  Combining the two statements, $\mathcal{A}$ is \emph{not} $\eps'$ for any $\eps' < \frac{1}{k} \log(\hat{p}_0/\hat{p}_1) = \eps_{\mathit{LB}}$.
\end{proof}

The $\epslb$ reported by Algorithm~\ref{alg:generic} has two fundamental upper bounds, the provable $\epsth$, and an upper bound, $\epsopt(T, \alpha)$, imposed by Monte Carlo estimation. The first upper bound is natural: if we run the algorithm on some $\mathcal{A}$ for which the $\eps$ we can prove is $\epsth=1$, then $\epslb\le \epsth=1$. To understand $\epsopt(T, \alpha)$, suppose we run 500 trials, and desire $\alpha=0.01$. The best possible performance is if we get perfect inference accuracy and $k=1$, where $\text{ct}_0=500$ and $\text{ct}_1=0$. The Clopper Pearson confidence interval produces $\hat{p}_0=0.989, \hat{p}_1=0.011$, which gives $\epslb=4.54/k=4.54$. Then, with $99\%$ probability, the true $\eps$ is at least $4.54$, and $\epsopt(T, \alpha)=4.54$.

We remark that the above procedure only demonstrates that $\mathcal{A}$ cannot be strictly better than $(\epslb,0)$-DP, but allows for it to be  $(\epslb/2, \delta)$-DP for very small $\delta$. However, in our work, $\hat{p}_0,\hat{p}_1$ turn out never to be too close to $0$, so these differences have little effect on our findings.  In Appendix
\ref{app:deltage0},
we formally discuss how to modify this algorithm for $(\eps,\delta)$-DP for $\delta > 0$. We also show when we can increase $\epslb$ by considering the maximum upper bounds of the original output set $\mathcal{O}$ and its complement $\mathcal{O}^C$.

\section{Poisoning Attacks}
\label{sec:poisoning}

\begin{figure}
    \centering
    \includegraphics[scale=.45]{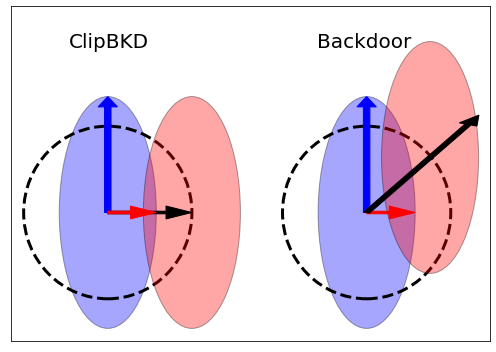}
    \caption{The distribution of gradients from an iteration of DP-SGD under a clean dataset (blue ellipse) and a poisoned dataset (red ellipse).  The right pair depicts traditional backdoors while the left pair depicts our backdoors.  Our attack pushes in the direction of least variance, so is impacted less by gradient clipping, which is indicated by the two distributions overlapping less.}
    \label{fig:clipbkd_diagram}
\end{figure}

We now show how to use poisoning attacks to run Algorithm~\ref{alg:generic}. Intuitively, we begin with a dataset $D_0$ and replace $k$ rows with poisoning points to form $D_1$; we then use the impact of poisoning as an output set $\mathcal{O}$. We start with existing backdoor attacks~\cite{gu2017badnets}, and then propose a more effective clipping-aware poisoning attack.

\subsection{Poisoning Background}
\label{ssec:pois_bg}
In a poisoning attack, an adversary replaces $k$ data points from a training dataset $D$ of $n$ points. The poisoned training dataset is provided as input to the training algorithm, which releases a model $f$ that minimizes a loss function $\mathcal{L}(f, D)$ on its given dataset $D$.


We focus on a specific type of poisoning attack, called a \emph{backdoor attack}~\cite{gu2017badnets}. In a backdoored model, the performance on natural data is maintained, but, by adding a small perturbation to a data point $x$ into $\mathit{Pert}(x)$, the adversary changes the predicted class of the perturbed data. These attacks have been developed for image datasets. 
In the original attack~\cite{gu2017badnets}, described in Algorithm~\ref{alg:backdoors}, the perturbation function $\mathit{Pert}(\cdot)$ adds a pattern in the corner of an image. The poisoning attack takes natural data $(x,y)$, perturbs the image to $\mathit{Pert}(x)$, and changes the class to some $y_p$. The objective is to decrease the loss on $(\mathit{Pert}(x), y_p)$ values from the perturbed test set.


\begin{algorithm}[t!]
\SetAlgoLined
\DontPrintSemicolon
\KwData{Dataset $X, Y$, poison size $k$, perturbation function $\mathit{Pert}$, target class $y_p$}
\SetKwFunction{Backdoor}{\textsc{Backdoor}}
\SetKwFunction{BackdoorTest}{\textsc{BackdoorTest}}

\SetKwProg{Fn}{Function}{:}{}
\Fn{\Backdoor{$X, Y, k, \mathit{Pert}, y_p$}}{
    $X_p =\textsc{GetRandomRows}(X, k)$\;
    $X'_p = \mathit{Pert}(X_p)$\;
    $X_{tr}^p = \textsc{ReplaceRandomRows}(X, X'_p)$\;
    $Y_{tr}^p = \textsc{ReplaceRandomRows}(Y, y_p)$\;
    \KwRet $D_0 = (X, Y), D_1 = (X_{tr}^p, Y_{tr}^p)$\;
}
\KwData{Model $f$, dataset $(X, Y)$, pert.~function $\mathit{Pert}$, target class $y_p$, loss function $\ell$, threshold $Z$}
\SetKwProg{Fn}{Function}{:}{}
\Fn{\BackdoorTest{$f, X, Y, \mathit{Pert}, y_p, \ell, Z$}}{
    $X_p = \mathit{Pert}(X)$\;
    \lIf{$\sum_{x_p\in X_p}\ell(f; (x_p, y_p)>Z$}{\Return Backdoored}
    \Return Not Backdoored
}
\caption{Baseline Backdoor Poisoning Attack and Test Statistic (Section~\ref{ssec:pois_bg})}
\label{alg:backdoors}
\end{algorithm}

\begin{algorithm}[t!]
\SetAlgoLined
\DontPrintSemicolon
\SetKwFunction{ClipBkd}{\textsc{ClipBkd}}
\SetKwFunction{ClipBkdTest}{\textsc{ClipBkdTest}}

\KwData{Dataset $X, Y$, pretrained model $f$, poison size $k$, dataset dimension $d$, norm $m$}
\SetKwProg{Fn}{Function}{:}{}

\Fn{\ClipBkd{$X, Y, k, f, m$}}{
    $U, D, V = \textit{SVD}(X)$\Comment{Singular value decomposition}
    
    $x_p = mV_d$ \Comment{$V_d$ is the singular vector for smallest singular value}
    
    $y_p = \argmin_i f(x_p)$ \Comment{Pick class maximizing gradient norm}
    
    $X_{tr}^p = \textsc{ReplaceRandomRows}(X, [x_p]*k)$ \Comment{Add poisoning point $k$ times}
    
    $Y_{tr}^p = \textsc{ReplaceRandomRows}(Y, [y_p]*k)$ \Comment{Add targeted class $k$ times}
    
    \KwRet $D_0 = (X, Y), D_1 = (X_{tr}^p, Y_{tr}^p)$
}
\KwData{Model $f$, Poison Data $x_p, y_p$, Threshold $Z$}
\SetKwProg{Fn}{Function}{:}{}

\Fn{\ClipBkdTest{$f, x_p, y_p, Z$}}{
    \lIf{$(f(x_p)-f(0^d))\cdot y_p>Z$}{\Return Backdoored}
    \Return Not Backdoored
}
\caption{Clipping-Aware Backdoor Poisoning Attack and Test Statistic (Section~\ref{ssec:clipbkd})}
\label{alg:bkdattack}
\end{algorithm}

\subsection{Clipping-Aware Poisoning}
\label{ssec:clipbkd}
DP-SGD makes two modifications to the learning process to preserve privacy: clipping gradients and adding noise.
Clipping provides no formal privacy on its own, but many poisoning attacks perform significantly worse in the presence of clipping.
Indeed, the basic backdoor attack from Section~\ref{ssec:pois_bg} results in a fairly weak lower bound of at most $\epslb=0.11$ using the Fashion-MNIST dataset, even with no added noise (which has an $\epsth=\infty$).
To improve this number, we must make the poisoning attack sufficiently robust to clipping.

To understand existing backdoor attacks' difficulty with clipping, consider clipping's impact on logistic regression. The gradient of model parameters $w$ with respect to a poisoning point $(x_p, y_p)$ is
$$\nabla_w \ell(w,b; x_p,y_p) = \ell'(w\cdot x_p + b, y_p)x_p.$$
Standard poisoning attacks, including the backdoor attack from Section~\ref{ssec:pois_bg}, focus on increasing $|\ell'(w\cdot x_p + b, y_p)|$; by doubling this quantity, if $|x_p|$ is fixed, half as many poisoning points are required for the same effect. However, in the presence of clipping, this relationship is broken.

To be more effective in the presence of clipping, the attack must produce not only large gradients, but \emph{distinguishable} gradients.  That is, the distribution of gradients arising from poisoned and cleaned data must be significantly different.  To analyze distinguishability, we consider the variance of gradients, illustrated in Figure~\ref{fig:clipbkd_diagram}, and seek a poisoning point $(x_p, y_p)$ minimizing $\mathit{Var}_{(x, y)\in D}\left[\nabla_w\ell(w, b; x_p, y_p)\cdot \nabla_w\ell(w, b; x, y)\right]$. This is dependent on the model parameters at a specific iteration of DP-SGD: we circumvent this issue by minimizing the following upper bound, which holds for all models (for logistic regression, $|\ell'(w\cdot x + b; y)| \le 1$):
$$\mathit{Var}_{(x, y)\in D}\left[\ell'(w\cdot x_p + b, y_p)x_p\cdot \ell'(w\cdot x + b, y)x\right]\le \mathit{Var}_{(x, y)\in D}\left[x_p\cdot x\right].$$

We can minimize this variance with respect to the poisoning point $x_p$ by using the singular value decomposition: selecting $x_p$ as the singular vector corresponding to the smallest singular value (i.e.~the direction of least variance), and scale $x_p$ to a similar norm to the rest of the dataset. We select $y_p$ to be the smallest probability class on $x_p$. We then insert $k$ copies of the poisoning point $(x_p,y_p)$.  We call this approach \textsc{ClipBKD}, detailed in Algorithm~\ref{alg:bkdattack}. We prove in Appendix 
\ref{app:clipbkdtheory}
that when we run \textsc{ClipBKD} (modified for regression tasks) to estimate the privacy of the \emph{output perturbation} algorithm, we obtain $\epslb$ within a small factor of the upper bound $\epsth$, giving evidence that this attack is well suited to our application in differential privacy. In Appendix
\ref{app:pretrained},
we describe how to adapt \textsc{ClipBKD} to transfer learning from a pre-trained model.

For both standard and clipping-aware backdoors, we generate $D_0, D_1$ with a given poisoning size $k$, using functions \textsc{Backdoor} or \textsc{ClipBkd}, respectively. Then the test statistic is whether the backdoored points are distinguishable by a threshold on their loss (i.e., output set $\mathcal{O}$ is whether \textsc{BkdTest} or \textsc{ClipBkdTest} return ``Backdoored''). We first run an initial phase of $T$ trials  to find a good threshold $Z$ for the test functions. We then run another $T$ trials in Algorithm~\ref{alg:generic} to estimate $\hat{p_0}$ and $\hat{p_1}$ based on either the \textsc{BkdTest} or the \textsc{ClipBkdTest} test statistic.

\section{Experiments and Discussion}
\label{sec:exp}
\subsection{Experimental Setup}
We evaluate both membership inference (MI, as used by \cite{yeom2018privacy} and \cite{jayaraman2019evaluating} and described in Appendix
\ref{app:mi})
and our algorithms on three datasets: Fashion-MNIST (FMNIST), CIFAR10, and Purchase-100 (P100). For each dataset, we consider both a \textbf{logistic regression (LR)} model and a \textbf{two-layer feedforward neural network (FNN)}, trained with DP-SGD using various hyperparameters:

\textbf{FMNIST}~\cite{xiao2017/online} is a dataset of 70000 28x28 pixel images of clothing from one of 10 classes, split into a train set of 60000 images and a test set of 10000 images. It is a standard benchmark dataset for differentially private machine learning. To improve training speed, we consider a simplified version, using only classes 0 and 1 (T-shirt and trouser), and downsample so each class contains 3000 training and 1000 testing points.
\textbf{CIFAR10}~\cite{Krizhevsky09learningmultiple} is a harder dataset than FMNIST, consisting of 60000 32x32x3 images of vehicles and animals, split into a train set of 50000 and a test set of 10000. For training speed, we again take only class 0 and 1 (airplane and automobile), making our train set contain 10000 samples, and the test set 2000 samples. When training on CIFAR10, we follow standard practice for differential privacy and fine-tune the last layer of a model pretrained nonprivately on the more complex CIFAR100, a similarly sized but more complex benchmark dataset~\cite{papernot2020making}.
\textbf{P100}~\cite{shokri2017membership} is a modification of a Kaggle dataset~\cite{Purchase100}, with 200000 records of 100 features, and 100 classes. The features are purchases, and the classes are user clusters. Following \cite{jayaraman2019evaluating}, we subsample the dataset so it has 10000 train records and 10000 test records.

Our techniques are general, and could be applied to any dataset-model pair to identify privacy risks for DP-SGD.  Examining these six dataset-model pairs demonstrates that our technique can be used to identify new privacy risks in DP-SGD, and a comprehensive empirical study is not our focus. 

\subsubsection{Implementation Details.}

\mypar{Model Size.} The two-layer feedforward neural networks all have a width of 32 neurons. For CIFAR10, the logistic regression model and feedforward neural network are added on top of the pretrained convolutional neural network.

\mypar{Computing Thresholds.} In order to run Algorithm~\ref{alg:generic}, we need to specify $D_0, D_1$ and $\mathcal{O}$. We've described how to use poisoning to compute $D_0, D_1$, and how the test statistics for these attacks are constructed, assuming a known threshold. To produce this threshold, we train 500 models on the unpoisoned dataset and 500 models on the poisoned dataset, and identify which of the resulting 1000 thresholds produces the best $\epslb$, using Algorithm~\ref{alg:generic}.

\mypar{Training Details.} We discuss the details of training in Table~\ref{tab:trainingdetails}. We selected these values to ensure a good tradeoff between accuracy and $\eps$, and selecting $\ell_2$ regularization for P100 based on \cite{jayaraman2019evaluating}.

\begin{table}[]
    \centering
    \begin{tabular}{|c|cccc|}
        \hline
        Dataset & Epochs & Learning Rate & Batch Size & $\ell_2$ Regularization \\
        \hline
        FMNIST & 24 & 0.15 & 250 & 0 \\
        CIFAR10 & 20 & 0.8 & 500 & 0 \\
        P100 & 100 & 2 & 250 & $10^{-4}$ / $10^{-5}$ \\
        \hline
    \end{tabular}
    \caption{Training details for experiments in Section~\ref{sec:exp}. P100 regularization is $10^{-5}$ for logistic regression and $10^{-4}$ for neural networks, following \cite{jayaraman2019evaluating}.}
    \label{tab:trainingdetails}
\end{table}

\subsection{Results and Discussion}
\begin{figure}[t]
    \centering
\begin{subfigure}{0.33\textwidth}
  \includegraphics[width=\linewidth]{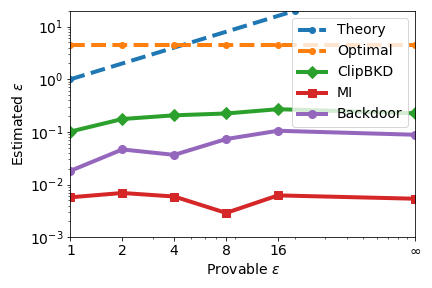}
  \caption{FMNIST, LR}
  \label{fig:fmnist_lr}
\end{subfigure}\hfil
\begin{subfigure}{0.33\textwidth}
  \includegraphics[width=\linewidth]{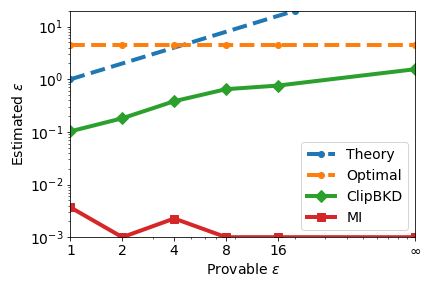}
  \caption{CIFAR10, LR}
  \label{fig:cifar_lr}
\end{subfigure}\hfil
\begin{subfigure}{0.33\textwidth}
  \includegraphics[width=\linewidth]{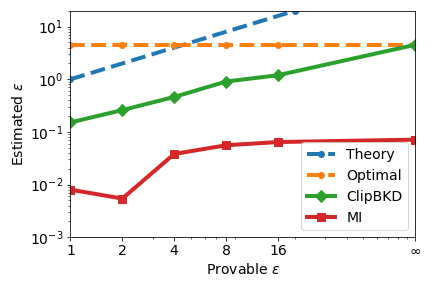}
  \caption{P100, LR}
  \label{fig:p100_lr}
\end{subfigure}

\medskip
\begin{subfigure}{0.33\textwidth}
  \includegraphics[width=\linewidth]{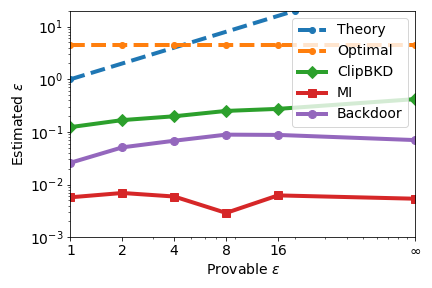}
  \caption{FMNIST, FNN}
  \label{fig:fmnist_2f}
\end{subfigure}\hfil
\begin{subfigure}{0.33\textwidth}
  \includegraphics[width=\linewidth]{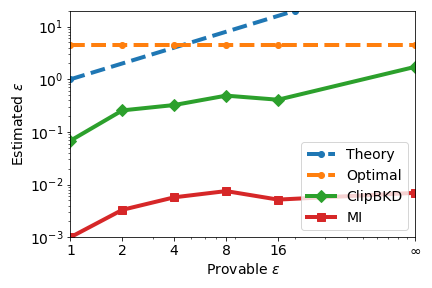}
  \caption{CIFAR10, FNN}
  \label{fig:cifar_2f}
\end{subfigure}\hfil
\begin{subfigure}{0.33\textwidth}
  \includegraphics[width=\linewidth]{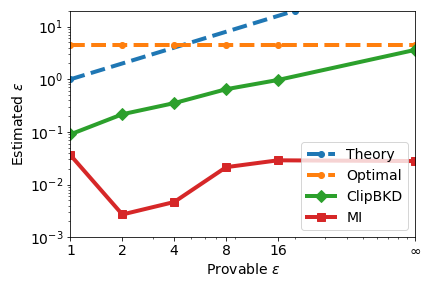}
  \caption{P100, FNN}
  \label{fig:p100_2f}
\end{subfigure}
\caption{Performance of privacy attacks MI, Backdoor, and ClipBKD on our datasets. LR = logistic regression, FNN = two-layer neural network. Backdoor attacks have not been developed for Purchase-100, so only MI and Clip-BKD were run. Backdoors do not provide positive $\epslb$ on CIFAR10 due to difficulty with the pretrained model.}
\label{fig:images}
\end{figure}


Figure~\ref{fig:images} presents a direct comparison of the privacy bounds produced by ClipBKD (our attack), the standard backdoor attack, and MI.  As standard backdoor attacks only exist for images, we only report results on them on FMNIST and CIFAR10. The pattern we choose for backdoor attacks is a 5x5 white square in the top-left corner of an image. For ClipBKD, we use $T = 500$ trials and confidence level $\alpha = 0.01$ (i.e., our Monte Carlo estimates hold with 99\% confidence) and report the best result from $k=1,2,4,8$ poisoning points.  Results for MI use 1000 samples, and average over 10 trained models.  For context, we display the best theoretical upper bound on $\eps_{\mathit{th}}$ and also $\eps_{\mathit{OPT}}(500, 0.01)$, which is the best value of $\eps_{\mathit{LB}}$ that we could hope to produce using $T$ trials and confidence level $\alpha$.

For every dataset and model, we find that ClipBKD significantly outperforms MI, by a factor of 2.5x--1500x. As a representative example, for $\eps_{\mathit{th}}=4$ on Purchase-100 with 2-layer neural networks, ClipBKD gives an $\eps_{\mathit{LB}}$ of 0.46, while MI gives $\eps_{\mathit{LB}}$ of 0.04, an improvement of 12.1x. We also find ClipBKD always improves over standard backdoors: on FMNIST by an average factor of 3.84x, and standard backdoors never reach positive $\epslb$ on CIFAR, due to the large number of points required to poison the pretrained model. We also find that ClipBKD returns $\epslb$ that are close to $\epsth$; for finite $\epsth$, the majority of gaps are a factor of $<12.3$x, reaching as low as 6.6x. For example, on Purchase-100, when $\epsth=4$, we find that ClipBKD returns an $\epslb$ of 0.46, a gap of 8.7x. 

\mypar{Sensitivity to Hyperparameters.}
We also give a more thorough evaluation of ClipBKD's performance as a function of DP-SGD's hyperparameters. We vary clipping norm between 0.5, 1, and 2 and vary the noise to ensure $\eps_{\mathit{th}}$ between 1, 2, 4, 8, 16, and $\infty$.  We also vary the initialization randomness between random normal initializations with variance $0$ (fixed initialization), $0.5\sigma$, $\sigma$, and $2\sigma$, where $\sigma$ is the variance of Glorot normal initialization.  Table~\ref{tab:nci} reports the best $\eps_{\mathit{LB}}$ produced by the attack over $k=1,2,4,8$.  Our best measured values of $\eps_{\mathit{LB}}$ occur when initialization is fixed, and are within a 4.2-7.7x factor of $\eps_{\mathit{th}}$, speaking to the effectiveness of ClipBKD. When $\epsth=\infty$ and the initialization is fixed, we achieve perfect inference accuracy, matching $\eps_{OPT}(500, 0.01)=4.54$.


These experiments reveal three intuitive trends. First, as $\epsth$ increases (equivalently, the noise level decreases), $\epslb$ also increases.  Second, as the initialization randomness decreases, $\epslb$ increases.  All existing analyses of DP-SGD give privacy upper bounds for \emph{any fixed initialization}, and our results suggest that initial randomization might play a significant role.  Finally, as clipping norm decreases, $\epslb$ decreases, except when the initialization is fixed. In fact, our results show that $\epslb$ is more sensitive to the clipping norm than the amount of noise.  All existing analyses of DP-SGD consider only the \emph{noise multiplier} $\sigma_{\mathit{GD}}$ but not the clipping norm, but the role of the clipping norm itself seems highly significant.  

We emphasize that for \emph{every} choice of hyperparameters, the training accuracy is 96--98\%, so the algorithm has comparable utility, but potentially very different privacy and robustness to poisoning, as we vary these parameters.  We believe these phenomena deserve further study.

\begin{table*}[t]
    \small
    \centering
    \begin{tabular}{c|cccc}
        \hline
        Params & Fixed Init & Init Rand = $0.5\sigma$ & Init Rand = $\sigma$ & Init Rand = $2\sigma$ \\
        \hline
        $\eps_{\mathit{th}}=1, \sigma_{\mathit{GD}}=5.02$      & 0.13 / 0.15 / 0.13 & 0.13 / 0.17 / 0.13 & 0.06 / 0.12 / 0.09 & 0.01 / 0.06 / 0.08 \\
        $\eps_{\mathit{th}}=2, \sigma_{\mathit{GD}}=2.68$      & 0.33 / 0.37 / 0.28 & 0.27 / 0.33 / 0.39 & 0.10 / 0.17 / 0.27 & 0.01 / 0.06 / 0.17 \\
        $\eps_{\mathit{th}}=4, \sigma_{\mathit{GD}}=1.55$      & 0.89 / 0.75 / 0.71 & 0.28 / 0.52 / 0.78 & 0.08 / 0.20 / 0.54 & 0.02 / 0.10 / 0.18 \\
        $\eps_{\mathit{th}}=8, \sigma_{\mathit{GD}}=1.01$      & 1.61 / 1.85 / 1.90 & 0.33 / 0.55 / 1.27 & 0.07 / 0.25 / 0.53 & 0.01 / 0.05 / 0.20 \\
        $\eps_{\mathit{th}}=16, \sigma_{\mathit{GD}}=0.73$     & 2.15 / 2.16 / 2.43 & 0.36 / 0.80 / 1.39 & 0.13 / 0.27 / 0.72 & 0.02 / 0.08 / 0.16  \\
        $\eps_{\mathit{th}}=\infty, \sigma_{\mathit{GD}}=0$    & 4.54 / 4.54 / 4.54 & 0.29 / 0.95 / 2.36 & 0.10 / 0.42 / 0.79 & 0.03 / 0.09 / 0.27 \\
    \end{tabular}
    \caption{Lower bound $\epslb$ measured with \textsc{ClipBKD} for clipping norms of (0.5 / 1 / 2) for two-layer neural networks trained on FMNIST. Training accuracy for all models is 96\%-98\%. Results are the maximum over $k=1, 2, 4, 8$. $\sigma_{\mathit{GD}}$ refers to the DP-SGD noise multiplier, while $\sigma$ is Glorot initialization randomness~\cite{glorot2010understanding}.  All reported values of $\epslb$ are valid with 99\% confidence over the randomness of our experiments.}
    \label{tab:nci}
\end{table*}

We present the same experiment in Table~\ref{tab:p100_nci}, run on the P100 dataset. In P100, we use $\ell_2$ regularization, a higher learning rate, and more epochs, making the contribution of the initialization smaller. As such, we find the role of both clipping norm and random initialization to be diminished. As a result, we manage to achieve $\epsopt(500, 0.01)$ without a fixed initialization.
\begin{table*}[t]
    \small
    \centering
    \begin{tabular}{c|ccc}
        \hline
        Params & Init Rand = $0.5\sigma$ & Init Rand = $\sigma$ & Init Rand = $2\sigma$ \\
        \hline
        $\eps_{\mathit{th}}=1, \sigma_{\mathit{GD}}=7.78$      & 0.09 / 0.01 / 0.00 & 0.05 / 0.00 / 0.00 & 0.07 / 0.05 / 0.00  \\
        $\eps_{\mathit{th}}=2, \sigma_{\mathit{GD}}=4.04$      & 0.16 / 0.27 / 0.11 & 0.21 / 0.17 / 0.03 & 0.20 / 0.10 / 0.02 \\
        $\eps_{\mathit{th}}=4, \sigma_{\mathit{GD}}=2.20$      & 0.38 / 0.33 / 0.30 & 0.29 / 0.35 / 0.30 & 0.34 / 0.33 / 0.13  
 \\
        $\eps_{\mathit{th}}=8, \sigma_{\mathit{GD}}=1.31$      & 0.52 / 0.53 / 0.42 & 0.54 / 0.53 / 0.52 & 0.56 / 0.46 / 0.50 \\
        $\eps_{\mathit{th}}=16, \sigma_{\mathit{GD}}=0.89$     & 0.80 / 0.77 / 0.71 & 0.63 / 0.77 / 0.76 & 0.74 / 0.70 / 0.72 
  \\
        $\eps_{\mathit{th}}=\infty, \sigma_{\mathit{GD}}=0$    & 2.73 / 4.53 / 4.54 & 1.52 / 3.08 / 4.52 & 0.90 / 1.91 / 2.79  \\
    \end{tabular}
    \caption{Lower bound $\epslb$ measured with \textsc{ClipBKD} for clipping norms of (0.5 / 1 / 2) for two-layer neural networks trained on P100. Results are the maximum over $k=1, 2, 4, 8$. $\sigma_{\mathit{GD}}$ refers to the DP-SGD noise multiplier, while $\sigma$ is Glorot initialization randomness~\cite{glorot2010understanding}. We do not run experiments with fixed initialization as we already achieve $\epsopt(500, 0.01)$ with initialization of $0.5\sigma$. All reported values of $\epslb$ are valid with 99\% confidence over the randomness of our experiments.}
    \label{tab:p100_nci}
\end{table*}

\section{Conclusion and Future Directions}
We use novel poisoning attacks to establish strong limits on the privacy of specific differentially private algorithms, namely DP-SGD.  We establish that the worst-case privacy bounds for this algorithm are approaching their limits.  Our findings highlight several questions for future exploration: 
\begin{itemize}
    \item How much can our attacks be pushed quantitatively?  Can the gap between our lower bounds and the worst-case upper bounds be closed?
    \item Can we incorporate additional features into the privacy analysis of DP-SGD, such as the specific gradient-clipping norm, and the amount of initial randomness?
    \item How realistic are the instances produced by our attacks, and can we extend the attacks to give easily interpretable examples of privacy risks for non-experts?
\end{itemize}

Although there is no hope of determining the precise privacy level of a given algorithm in a fully empirical way, we believe our work demonstrates how a quantitative, empirical approach to privacy attacks can effectively complement analytical work on privacy in machine learning.

\section*{Acknowledgments}
JU is supported by NSF grants CCF-1750640, CNS-1816028, and CNS-1916020, and a Google Faculty Research Award.  This research was  also sponsored by a Toyota ITC research award and the U.S. Army Combat Capabilities Development
Command Army Research Laboratory  under Cooperative Agreement Number W911NF-13-2-0045 (ARL Cyber Security CRA). The views and conclusions contained in this document are those of the authors and
should not be interpreted as representing the official policies, either expressed or implied, of the Combat Capabilities Development Command Army Research Laboratory or the U.S. Government. The U.S. Government is authorized to reproduce and distribute reprints for Government purposes notwithstanding any
copyright notation here on.

\bibliographystyle{alpha}
\bibliography{references-matthew,references-jon}

\newcommand{\etalchar}[1]{$^{#1}$}
\begin{thebibliography}{DWW{\etalchar{+}}18}

\bibitem[ACG{\etalchar{+}}16]{Abadi+16}
Martin Abadi, Andy Chu, Ian Goodfellow, H~Brendan McMahan, Ilya Mironov, Kunal
  Talwar, and Li~Zhang.
\newblock Deep learning with differential privacy.
\newblock In {\em ACM Conference on Computer and Communications Security},
  CCS'16, 2016.

\bibitem[BDRS18]{BunDRS18}
Mark Bun, Cynthia Dwork, Guy~N Rothblum, and Thomas Steinke.
\newblock Composable and versatile privacy via truncated cdp.
\newblock In {\em Proceedings of the 50th Annual ACM SIGACT Symposium on Theory
  of Computing}, pages 74--86, 2018.

\bibitem[BEM{\etalchar{+}}17]{Bittau+17}
Andrea Bittau, {\'U}lfar Erlingsson, Petros Maniatis, Ilya Mironov, Ananth
  Raghunathan, David Lie, Mitch Rudominer, Ushasree Kode, Julien Tinnes, and
  Bernhard Seefeld.
\newblock Prochlo: Strong privacy for analytics in the crowd.
\newblock In {\em Proceedings of the 26th Annual ACM Symposium on Operating
  Systems Principles}, SOSP '17, pages 441--459, Shanghai, China, 2017. ACM.

\bibitem[BS16]{BunS16}
Mark Bun and Thomas Steinke.
\newblock Concentrated differential privacy: Simplifications, extensions, and
  lower bounds.
\newblock In {\em Theory of Cryptography Conference}, pages 635--658. Springer,
  2016.

\bibitem[BST14]{BassilyST14}
Raef Bassily, Adam Smith, and Abhradeep Thakurta.
\newblock Private empirical risk minimization: Efficient algorithms and tight
  error bounds.
\newblock In {\em Proceedings of the 55th IEEE Annual Symposium on Foundations
  of Computer Science}, FOCS '14, pages 464--473, Philadelphia, PA, 2014. IEEE.

\bibitem[CP34]{clopper1934use}
Charles~J Clopper and Egon~S Pearson.
\newblock The use of confidence or fiducial limits illustrated in the case of
  the binomial.
\newblock {\em Biometrika}, 26(4):404--413, 1934.

\bibitem[DKM{\etalchar{+}}06]{DworkKMMN06}
Cynthia Dwork, Krishnaram Kenthapadi, Frank McSherry, Ilya Mironov, and Moni
  Naor.
\newblock Our data, ourselves: Privacy via distributed noise generation.
\newblock In {\em 25th Annual International Conference on the Theory and
  Applications of Cryptographic Techniques}, EUROCRYPT '06, pages 486--503,
  2006.

\bibitem[DMNS06]{DworkMNS06}
Cynthia Dwork, Frank McSherry, Kobbi Nissim, and Adam Smith.
\newblock Calibrating noise to sensitivity in private data analysis.
\newblock In {\em Proceedings of the 3rd Conference on Theory of Cryptography},
  TCC '06, pages 265--284, Berlin, Heidelberg, 2006. Springer.

\bibitem[DN03]{DinurN03}
Irit Dinur and Kobbi Nissim.
\newblock Revealing information while preserving privacy.
\newblock In {\em Proceedings of the 22nd ACM Symposium on Principles of
  Database Systems}, PODS '03, pages 202--210. ACM, 2003.

\bibitem[DR16]{DworkR16}
Cynthia Dwork and Guy~N Rothblum.
\newblock Concentrated differential privacy.
\newblock {\em arXiv preprint arXiv:1603.01887}, 2016.

\bibitem[DRS19]{DongRS19}
Jinshuo Dong, Aaron Roth, and Weijie~J Su.
\newblock Gaussian differential privacy.
\newblock {\em arXiv preprint arXiv:1905.02383}, 2019.

\bibitem[DSS{\etalchar{+}}15]{DworkSSUV15}
Cynthia Dwork, Adam Smith, Thomas Steinke, Jonathan Ullman, and Salil Vadhan.
\newblock Robust traceability from trace amounts.
\newblock In {\em IEEE Symposium on Foundations of Computer Science}, FOCS '15,
  2015.

\bibitem[DWW{\etalchar{+}}18]{DingWWZK18}
Zeyu Ding, Yuxin Wang, Guanhong Wang, Danfeng Zhang, and Daniel Kifer.
\newblock Detecting violations of differential privacy.
\newblock In {\em {ACM} {SIGSAC} Conference on Computer and Communications
  Security}, CCS'18, 2018.

\bibitem[EPK14]{ErlingssonPK14}
{\'U}lfar Erlingsson, Vasyl Pihur, and Aleksandra Korolova.
\newblock {RAPPOR}: Randomized aggregatable privacy-preserving ordinal
  response.
\newblock In {\em ACM Conference on Computer and Communications Security}, CCS
  '14, 2014.

\bibitem[GB10]{glorot2010understanding}
Xavier Glorot and Yoshua Bengio.
\newblock Understanding the difficulty of training deep feedforward neural
  networks.
\newblock In {\em Proceedings of the thirteenth international conference on
  artificial intelligence and statistics}, pages 249--256, 2010.

\bibitem[GDGG17]{gu2017badnets}
Tianyu Gu, Brendan Dolan-Gavitt, and Siddharth Garg.
\newblock Badnets: Identifying vulnerabilities in the machine learning model
  supply chain.
\newblock {\em arXiv preprint arXiv:1708.06733}, 2017.

\bibitem[GM18]{GilbertM18}
Anna~C Gilbert and Audra McMillan.
\newblock Property testing for differential privacy.
\newblock In {\em 2018 56th Annual Allerton Conference on Communication,
  Control, and Computing (Allerton)}, pages 249--258. IEEE, 2018.

\bibitem[Goo]{TFP}
Google.
\newblock {Tensorflow Privacy}.

\bibitem[HCK{\etalchar{+}}20]{hong2020effectiveness}
Sanghyun Hong, Varun Chandrasekaran, Yi{\u{g}}itcan Kaya, Tudor Dumitra{\c{s}},
  and Nicolas Papernot.
\newblock On the effectiveness of mitigating data poisoning attacks with
  gradient shaping.
\newblock {\em arXiv preprint arXiv:2002.11497}, 2020.

\bibitem[HMA{\etalchar{+}}17]{Haney+17}
Samuel Haney, Ashwin Machanavajjhala, John~M Abowd, Matthew Graham, Mark
  Kutzbach, and Lars Vilhuber.
\newblock Utility cost of formal privacy for releasing national
  employer-employee statistics.
\newblock In {\em Proceedings of the 2017 ACM International Conference on
  Management of Data}, SIGMOD '17, pages 1339--1354, Chicago, IL, 2017. ACM.

\bibitem[HSR{\etalchar{+}}08]{homer2008resolving}
Nils Homer, Szabolcs Szelinger, Margot Redman, David Duggan, Waibhav Tembe,
  Jill Muehling, John~V Pearson, Dietrich~A Stephan, Stanley~F Nelson, and
  David~W Craig.
\newblock Resolving individuals contributing trace amounts of dna to highly
  complex mixtures using high-density snp genotyping microarrays.
\newblock {\em PLoS genetics}, 4(8), 2008.

\bibitem[HT19]{HylandT19}
Stephanie~L Hyland and Shruti Tople.
\newblock On the intrinsic privacy of stochastic gradient descent.
\newblock {\em arXiv preprint arXiv:1912.02919}, 2019.

\bibitem[JE19]{jayaraman2019evaluating}
Bargav Jayaraman and David Evans.
\newblock Evaluating differentially private machine learning in practice.
\newblock In {\em 28th $\{$USENIX$\}$ Security Symposium ($\{$USENIX$\}$
  Security 19)}, pages 1895--1912, 2019.

\bibitem[Kri09]{Krizhevsky09learningmultiple}
Alex Krizhevsky.
\newblock Learning multiple layers of features from tiny images.
\newblock Technical report, 2009.

\bibitem[Mir17]{Mironov17}
Ilya Mironov.
\newblock R{\'e}nyi differential privacy.
\newblock In {\em 2017 IEEE 30th Computer Security Foundations Symposium
  (CSF)}, pages 263--275. IEEE, 2017.

\bibitem[MTZ19]{MironovTZ19}
Ilya Mironov, Kunal Talwar, and Li~Zhang.
\newblock R$\backslash$'enyi differential privacy of the sampled gaussian
  mechanism.
\newblock {\em arXiv preprint arXiv:1908.10530}, 2019.

\bibitem[MZH19]{ma2019data}
Yuzhe Ma, Xiaojin Zhu, and Justin Hsu.
\newblock Data poisoning against differentially-private learners: Attacks and
  defenses.
\newblock {\em arXiv preprint arXiv:1903.09860}, 2019.

\bibitem[PCS{\etalchar{+}}20]{papernot2020making}
Nicolas Papernot, Steve Chien, Shuang Song, Abhradeep Thakurta, and Ulfar
  Erlingsson.
\newblock Making the shoe fit: Architectures, initializations, and tuning for
  learning with privacy, 2020.

\bibitem[PSM{\etalchar{+}}18]{Papernot+18}
Nicolas Papernot, Shuang Song, Ilya Mironov, Ananth Raghunathan, Kunal Talwar,
  and Úlfar Erlingsson.
\newblock Scalable private learning with pate.
\newblock In {\em International Conference on Learning Representations},
  ICLR'18, 2018.

\bibitem[Pur]{Purchase100}
Acquire valued shoppers challenge.

\bibitem[SCS13]{SongCS13}
Shuang Song, Kamalika Chaudhuri, and Anand~D Sarwate.
\newblock Stochastic gradient descent with differentially private updates.
\newblock In {\em 2013 IEEE Global Conference on Signal and Information
  Processing}, pages 245--248. IEEE, 2013.

\bibitem[SOJH09]{SankararamanOJH09}
Sriram Sankararaman, Guillaume Obozinski, Michael~I Jordan, and Eran Halperin.
\newblock Genomic privacy and limits of individual detection in a pool.
\newblock {\em Nature genetics}, 41(9):965--967, 2009.

\bibitem[SSSS17]{shokri2017membership}
Reza Shokri, Marco Stronati, Congzheng Song, and Vitaly Shmatikov.
\newblock Membership inference attacks against machine learning models.
\newblock In {\em 2017 IEEE Symposium on Security and Privacy (SP)}, pages
  3--18. IEEE, 2017.

\bibitem[TVV{\etalchar{+}}17]{AppleDP}
A.G. Thakurta, A.H. Vyrros, U.S. Vaishampayan, G.~Kapoor, J.~Freudiger, V.R.
  Sridhar, and D.~Davidson.
\newblock Learning new words, 2017.
\newblock US Patent 9,594,741.

\bibitem[XRV17]{xiao2017/online}
Han Xiao, Kashif Rasul, and Roland Vollgraf.
\newblock {Fashion-MNIST}: a novel image dataset for benchmarking machine
  learning algorithms, 2017.

\bibitem[YGFJ18]{yeom2018privacy}
Samuel Yeom, Irene Giacomelli, Matt Fredrikson, and Somesh Jha.
\newblock Privacy risk in machine learning: Analyzing the connection to
  overfitting.
\newblock In {\em 2018 IEEE 31st Computer Security Foundations Symposium
  (CSF)}, pages 268--282. IEEE, 2018.

\bibitem[YLP{\etalchar{+}}19]{YuLPGT19}
Lei Yu, Ling Liu, Calton Pu, Mehmet~Emre Gursoy, and Stacey Truex.
\newblock Differentially private model publishing for deep learning.
\newblock In {\em 2019 IEEE Symposium on Security and Privacy (SP)}, pages
  332--349. IEEE, 2019.

\end{thebibliography}

\appendix
\section{Extending Algorithm~\ref{alg:generic}.}
\label{app:deltage0}
\paragraph{Measuring $\eps$ when $\delta>0$.}
Notice that Algorithm~\ref{alg:generic} holds for $(\eps, 0)$-differential privacy.
However, this is only for simplicity---the group privacy guarantee of $(\eps, \delta)$-differential privacy can be converted to a similar procedure.
Write $x=\exp(\eps)$ in the group privacy guarantee for Equation~\ref{eq:group}, and rearrange to provide $p_1x^{k+1}-(p_1-\delta)x^k - p_0x + (p_0-\delta) \ge 0$. We can solve for $x$ here using a root-finding algorithm to find $x$, and computing $\eps_{\mathit{LB}}=\ln(x)$. Theorem~\ref{thm:generic} can be easily extended to this case. 

\paragraph{Measuring $\eps_{\mathit{LB}}$ with both $\mathcal{O}$ and $\mathcal{O}^C$.}
Notice that differential privacy makes a guarantee for all output sets, including the complement $\mathcal{O}^C$; if $\Pr[A(D)\in\mathcal{O}]=p$, then $\Pr[A(D)\in\mathcal{O}^C]=1-p$. If, upon computing $p_0, p_1$, we can compute a larger $\eps_{\mathit{LB}}$ by using $\mathcal{O}^C$, this requires no extra trials.

For example, suppose $\delta=0$, $k=1$, $p_1=0.8$, and $p_0=0.4$. Here, $\eps_{LB}=\log(p_1/ p_0)/1=\log(2)$. If, instead, we replace $\mathcal{O}$ by $\mathcal{O}^C$, $\eps_{LB}=\log((1-p_0)/(1-p_1))/1=\log(0.6/0.2)=\log(3)$. In Lemma~\ref{thm:complement}, we show when this technique improves $\eps_{LB}$: when $p_1>p_0+k\delta$ and $p_0+p_1 > 1$. We use this modification in all of our experiments.

\begin{figure}[h]
    \centering
    \includegraphics[scale=.45]{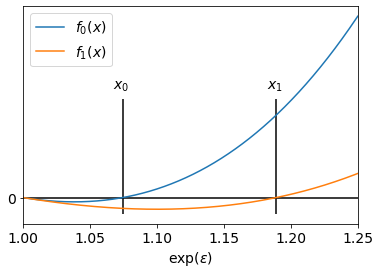}
    \caption{$f_0(x)$ and $f_1(x)$, as defined in Lemma~\ref{thm:complement} with $\delta=10^{-5}$, $k=4$, $p_0=0.6$, $p_1=0.8$.}
    \label{fig:complement}
\end{figure}
\begin{lemma}
If $p_1>p_0+k\delta$ and $p_0+p_1>1$, then the largest root of $f_0(x)=p_1x^{k+1}-(p_1-\delta)x^k - p_0x + (p_0-\delta)$ is smaller than the largest root of $f_1(x)=(1-p_0)x^{k+1}-(1-p_0-\delta)x^k - (1-p_1)x + (1-p_1-\delta)$.
\label{thm:complement}
\end{lemma}

\begin{proof}
Write $x_0$ the largest root of $f_0(x)$, and $x_1$ the largest root of $f_1(x)$. We show this in two parts: first, we show that, for all $p_0, p_1$, $f_0(x)$ has a root $x>1$ when $p_1>p_0+k\delta$, after which it is monotonically increasing. This shape is evident in Figure~\ref{fig:complement}. This provides a nonzero $\eps_{LB}$. Then we show that this $\eps_{LB}$ must be smaller for $f_1(x)$ when $p_0+p_1>1$, because $f_0(x)-f_1(x)>0$ when $x>1$. This ensures that $f_1(x_0)<0$, and so the root $x_1>x_0$.

We begin by showing that $f_0(x)$ has a single root $x>1$. First, notice that $f_0(1)=0$. We then analyze the derivative, showing that it starts negative, has a root, and then is always positive. This indicates that there can only be one root.
\[
f_0'(x) = (k+1)x^kp_0 + kx^{k-1}(\delta-p_0) - p_1=kp_0x^{k-1}(x-1) + x^kp_0 + k\delta x^{k-1} - p_1
\]
This has a root $x>1$ if $x^kp_0 + k\delta x^{k-1} - p_1<0$, so we require $p_0+k\delta - p_1 < 0$. Notice too that $f_0'(x)$ is monotonically increasing at $x>1$. This ensures that it has only one root $x>1$. This argument holds, too, for $f_1(x)$, as if $p_0+k\delta - p_1<0$, then $(1-p_1)+k\delta - (1-p_0)<0$.

Now that we know both $f_0$ and $f_1$ only have a single root, and they are both increasing at that root, we just need to show that $f_1(x_0)<0$, as this will ensure $x_1>x_0$. We do this by showing that $\forall x>1$, $f_0(x)-f_1(x)>0$. First, write 
\[
f_0(x)-f_1(x)=x^{k+1}(p_0+p_1-1)+x^k(1-p_0-p_1)+x(1-p_0-p_1)+p_0+p_1-1.
\]
The $\delta$ terms cancel, and we can factor the above into
\[
f_0(x)-f_1(x)=(p_0+p_1-1)(x-1)(x^k-1).
\]
This is always positive when $x>1$ and $p_0+p_1>1$.
\qedhere
\end{proof}

\section{Analysis of Backdoor Poisoning-based Auditing}
\label{app:clipbkdtheory}
We now provide formal evidence for the effectiveness of backdoor poisoning attacks in auditing differentially private algorithms with a case study on linear regression. To the best of our knowledge, this is also the first formal analysis of backdoor poisoning attacks for a concrete learning algorithm.

\begin{theorem}
Given a dataset $X\in \mathbb{R}^{n\times d}, Y\in [-.5, .5]^n$, where each $x_i\in X$ satisfies $|x_i|_2\le 1$, consider output perturbation to compute ridge regression with regularization $\lambda$, satisfying $\eps, \delta$ differential privacy. Then Algorithm~\ref{alg:simplelinregbackdoor} produces a backdoor attack that produces a lower bound $$\epslb = \frac{\lambda \eps}{(1+\lambda + \sigma_d)\sqrt{\pi\ln(1.25/\delta)}} - 4\delta$$ where $\sigma_d$ is the smallest singular value of $X$.
\end{theorem}
\jon{I think you want to use $\sigma_{\mathit{min}}$ for the least singular value.}
\jon{When you say it produces a backdoor attack that produces this lower bound, maybe you should clarify a bit?  You're saying that it leads to these $p_0,p_1$ values, which we then have to estimate.}

\begin{algorithm}[t]
\SetAlgoLined
\KwData{ Dataset $X\in \mathbb{R}^{n\times d}, Y\in [-.5, .5]$}
\KwResult{$D_c = (X_p, Y_p), D_p=(X_p, Y_p')$}

$U, D, V = \textit{SVD}(X)$\Comment{Singular value decomposition}

$x_p = V_d$

$y_p = .5$

\Return $(X\mid\mid x_p, Y\mid\mid y_p), (X\mid\mid x_p, Y\mid\mid -y_p)$
\caption{Clipping-Aware Poisoning Attack Generation}
 \label{alg:simplelinregbackdoor}
\end{algorithm}

\begin{proof}
We begin by computing the difference between the optimal linear regression parameters $w_0, w_1$ for the two datasets $D_0=(X\mid\mid x_p, Y\mid\mid y_p), D_1=(X\mid\mid x_p, Y\mid\mid -y_p)$, respectively. We refer to $X_p=(X\mid\mid x_p), Y_0=(Y\mid\mid y_p), Y_1=(Y\mid\mid -y_p)$. We continue to refer to the eigendecomposition of $X^T X=V D V^T$. Recall that the optimal parameters for an arbitrary dataset $X, Y$ with $\ell_2$ regularization $\lambda$ is $(\lambda I + X^T X)^{-1}X^T Y$.
\begin{align*}
w_0 - w_1 
={} &(\lambda I + X_p^T X_p)^{-1}X_p^T(Y_0 - Y_1) \\
={} &(\lambda I + X_p^T X_p)^{-1}x_p \\
={} &(\lambda I + VDV^T + v_dv_d^T)^{-1}v_d \\
={} &(\lambda V I V^T + V D V^T + V\text{diag}(e_d)V^T)^{-1}v_d \\
={} &V(\lambda I + D + \text{diag}(e_d))^{-1}V^Tv_d \\
={} &\frac{v_d}{\lambda + \sigma_d + 1}
\end{align*}

The output perturbation algorithm for Ridge regression, with $(\epsilon, \delta)$-DP, adds Gaussian noise with variance $\sigma^2 = 2\ln(1.25/\delta)(2/\lambda)^2/\epsilon^2$ to the optimal parameters $w$.

\jon{Use the standard notation $\mathcal{N}(\mu,\sigma^2)$.  Also you need to indicate that it's multivariate, usually with $\mathcal{N}(\mu,\sigma^2 \mathbb{I})$ or something.}
The optimal distinguisher for $w_0+\mathcal{N}(0,\sigma^2\mathbb{I})$ and $w_1+\mathcal{N}(0,\sigma^2\mathbb{I})$ is 
\[
f(w)=\mathbbm{1}[w\cdot v_d - 0.5(w_0+w_1)\cdot v_d < 0].
\]
Letting $c=\frac{0.5}{\lambda + \sigma_d + 1}$, the probability of success for this distinguisher is
$
\Pr[\mathcal{N}(0, \sigma^2) < c],
$
which gives an $\epsilon$ lower bound of
\[
\ln\left(\frac{\Pr[\mathcal{N}(0, \sigma^2) < c] - \delta}{\Pr[\mathcal{N}(0, \sigma^2) < - c]}\right).
\]
We can lower bound $\Pr[0<\mathcal{N}(0, \sigma^2)<c]$ using the following integral approximation:
\[
\Pr[0<\mathcal{N}(0, \sigma^2)<c]\ge \frac{c}{\sigma\sqrt{2\pi}}\exp(-c^2/2\sigma^2),
\] so our $\eps$ lower bound is
\[
\ln\left(\frac{0.5 - \delta + c/(\sigma\exp(c^2/2\sigma^2)\sqrt{2\pi})}{0.5 - c/(\sigma\exp(c^2/2\sigma^2)\sqrt{2\pi})}\right)\ge \ln\left(\frac{0.5 - \delta + c/(\sigma\exp(c^2/2\sigma^2)\sqrt{2\pi})}{0.5 + \delta - c/(\sigma\exp(c^2/2\sigma^2)\sqrt{2\pi})}\right).
\]
By its Maclaurin series, $\ln\left(\frac{0.5+x}{0.5-x}\right)\ge 4x$. Then we can compute our lower bound on $\epsilon$ to be
\[
\epslb \ge 4\left(\frac{c}{\sigma\exp(c^2/2\sigma^2)\sqrt{2\pi}} - \delta\right) =  O\left(\frac{\lambda\eps}{(1+\lambda+\sigma_d)\sqrt{\ln(1/\delta)}}\right)
\]
so the attack differs by a constant factor from the provable $\epsilon$. \jon{In the theorem you make this a firm lower bound but here you use $O(\cdot)$ notation.  Which is it?}\matthew{I'm using both here, it's just c is a bit complicated to write down so I'm using O() for that}

\end{proof}

\section{ClipBKD with Pretrained Models}
\label{app:pretrained}
State-of-the-art differentially private CIFAR10 models use transfer learning from a fixed pretrained CIFAR100 convolutional neural network. We call the pretrained model function $f_0$, which is never updated during training. Training produces a $f_1$, such that the entire prediction model is $f(x)=f_1(f_0(x))$.

This makes ClipBKD not directly applicable, as ClipBKD requires access to the input of the trained model $f_1$. Then we must try to produce some $x$ such that $f_0(x)=h_p$, where $h_p$ is produced by SVD in Algorithm~\ref{alg:backdoors}. This is not in general possible, so we instead use gradient descent to optimize the combination of two loss functions on $x$.

Our first loss function incentivizes decreasing $h_p\cdot v$ for high-variance directions $v\in V_{high}$ from SVD. This ensures the gradient will not move in SGD's noisy directions. Our second loss function incentivizes increasing $h_p\cdot v$ for low-variance directions $v\in V_{low}$ from SVD, ensuring the gradient is distinguishable in low variance directions. Putting these together, we produce $x_p$ by optimizing the following loss function:
\begin{equation}
\begin{split}
x_p=&\argmax_{x} \sum_{v\in V_{low}}(f_0(x)\cdot v)^2 - \sum_{v\in V_{high}}(f_0(x)\cdot v)^2 \\
&\text{s.t.} ~~ x\in [0, 1]^d
\end{split}
\end{equation}

We perform this optimization by projected gradient descent, running 10000 iterations with a learning rate of 1.

\section{Membership Inference}
\label{app:mi}
In membership inference~\cite{yeom2018privacy}, an adversary is given a model $f$ and its training loss $c$ and seeks to understand whether a given data point $(x, y)$ was used to train the model. Although alternative formulations have been proposed~\cite{shokri2017membership}, we focus on the one proposed by \cite{yeom2018privacy}. Intuitively, the attack relies on a generalization gap: the loss on training data should be smaller than the loss on test data. The algorithm is provided a set of $2n$ elements, $n$ of which were used for training, and $n$ not used for training, and predicts that any sample with loss lower than the training loss. The accuracy of these predictions is bounded by $\tfrac{\exp(\eps)}{1+\exp(\eps)}$ for any $\eps$-differentially private algorithm, so we can produce a lower bound $\epslb$ from it. The algorithm is provided in Algorithm~\ref{alg:mi}.

\begin{algorithm}[t]
\SetAlgoLined
\KwData{Training dataset $D_{tr}$ of size $n$, Test dataset $D_{t}$ of size $n$, Dataset, learning algorithm $\mathcal{A}$}
\vspace{3mm}
$f, c = \mathcal{A}(D_{tr})$  \Comment{$\mathcal{A}$ returns model and training loss}

$\text{correct\_ct} = 0$

\For{$(x,y)\in D_{tr}$}{
    \lIf{$\ell(f; x, y)<c$}{$\text{correct\_ct} = \text{correct\_ct}+1$\Comment{training set should have small loss}}
}
\For{$(x, y)\in D_{t}$}{
    \lIf{$\ell(f; x, y)>c$}{$\text{correct\_ct} = \text{correct\_ct}+1$\Comment{test set should have higher loss}}
}

$\text{Adv} = \text{correct\_ct}/2n$

$\epslb = \ln(\text{Adv}/(1-\text{Adv}))$
\caption{Membership Inference~\cite{yeom2018privacy}}
 \label{alg:mi}
\end{algorithm}

\end{document}